\documentclass[a4paper]{article}
\pdfoutput=1
\usepackage[margin=25mm,nohead]{geometry}
\usepackage{graphicx,amssymb,amsthm,amsmath,color}
\usepackage{enumerate}
\usepackage{xspace}

\def\rho{\varrho}

\newtheorem{cl}{Claim}

\def\A{\ensuremath{\mathrm{A}}}
\def\B{\ensuremath{\mathrm{B}}}
\def\T{\ensuremath{\mathrm{T}}}
\def\C{\ensuremath{\mathrm{C}}}
\def\F{\ensuremath{\mathrm{F}}}
\def\P{\ensuremath{\mathrm{P}}}
\def\pst{\ensuremath{\mathrm{P}_{st}}}
\def\spt{\ensuremath{\mathrm{SpT}}}
\def\cut{\ensuremath{\mathrm{Cut}}}
\def\d{\textbf{\boldmath{\ensuremath{\mspace{2mu} \wedge \mspace{2mu}}}}}
\def\c{\textbf{\boldmath{\ensuremath{\mspace{2mu} \mathrm{\cup} \mspace{2mu}}}}}
\def\u{\textbf{\boldmath{\ensuremath{\mspace{1mu} + \mspace{1mu}}}}}
\def\Hc{Hamiltonian circuit}
\def\Hcs{Hamiltonian circuit }
\def\Hp{Hamiltonian path}

\def\iff{if and only if }

\newcommand{\stp}{$s$-$t$ path }

\newcommand{\cE}{\ensuremath{\mathcal E}}

\newtheorem{thm}{\bfseries Theorem}
\newtheorem{cor}[thm]{\bfseries Corollary}

\newtheorem{prob}[thm]{\bfseries Problem}

\newcommand{\NPC}{\textbf{NPC\xspace}}
\newcommand{\inP}{\textbf{P\xspace}}

\newcommand{\threesat}{\textsc{3SAT}}
\newcommand{\oneinthree}{\textsc{One-In-Three 3SAT}}
\newcommand{\ov}{\overline} 
\newcommand{\cC}{\ensuremath{\mathcal{C}}}

\begin{document}


\title{On the tractability of some natural\\
 packing, covering  and partitioning problems}


\author{Attila Bern\'ath\thanks{Hungarian Academy of Sciences, Institute for
    Computer Science and Control (MTA-SZTAKI), Budapest,
    Hungary. This work was supported
    by MTA-ELTE Egerv\'ary Research Group, by the OTKA grant CK80124
    and by the ERC StG project PAAl no. 259515.
E-mail: {\tt bernath@cs.elte.hu}.}
\and
Zolt\'an Kir\'aly\thanks{Department of Computer Science and 
Egerv\'ary Research Group (MTA-ELTE), E\"otv\"os University,  P\'azm\'any
P\'eter s\'et\'any 1/C, Budapest, Hungary.  
Research was supported by 
grants (no.\ CNK 77780 and no.\ K 109240) from the National Development
 Agency of Hungary, based on a source from the Research and Technology 
Innovation Fund.
E-mail: {\tt kiraly@cs.elte.hu}}
}

\maketitle

\abstract{ In this paper we fix 7 types of undirected graphs: paths,
  paths with prescribed endvertices, circuits, forests, spanning
  trees, (not necessarily spanning) trees and cuts.  Given an
  undirected graph $G=(V,E)$ and two ``object types'' $\A$ and $\B$
  chosen from the alternatives above, we consider the following
  questions.  \textbf{Packing problem:} can we find an object of type
  $\A$ and one of type $\B$ in the edge set $E$ of $G$, so that they
  are edge-disjoint?  \textbf{Partitioning problem:} can we partition
  $E$ into an object of type $\A$ and one of type $\B$?
  \textbf{Covering problem:} can we cover $E$ with an object of type
  $\A$, and an object of type $\B$?  This framework includes 44
  natural graph theoretic questions.  Some of these problems were
  well-known before, for example covering the edge-set of a graph with
  two spanning trees, or finding an $s$-$t$ path $P$ and an $s'$-$t'$
  path $P'$ that are edge-disjoint. However, many others were not, for
  example can we find an $s$-$t$ path $P\subseteq E $ and a spanning
  tree $T\subseteq E$ that are edge-disjoint?  Most of these
  previously unknown problems turned out to be NP-complete, many of
  them even in planar graphs. This paper determines the status of
  these 44 problems. For the NP-complete problems we also investigate
  the planar version, for the polynomial problems we consider the
  matroidal generalization (wherever this makes sense).}

\section{Introduction}

In this paper we consider undirected graphs. The node set of a graph
$G=(V,E)$ is sometimes also denoted by $V(G)$, and similarly, the edge
set is sometimes denoted by $E(G)$.  A \textbf{subgraph} of a graph
$G=(V,E)$ is a pair $(V', E')$ where $V'\subseteq V$ and $E'\subseteq
E\cap (V' \times V')$. 
A graph is called \textbf{subcubic} if every node is
incident to at most 3 edges,
 and it is called \textbf{subquadratic} if every node is incident to at most 4 edges.
By a \textbf{cut} in a graph we mean the set
of edges leaving a nonempty proper subset $V'$ of the nodes (note that
we do not require that $V'$ and $V-V'$ induces a connected graph).  We use standard terminology and
refer the reader to \cite{frankbook} for what is not defined here.

We consider 3 types of decision problems with 7 types of
objects. The three types of problems are: packing, covering and
partitioning, and the seven types of objects are the following: paths
(denoted by a $\P$), paths with specified endvertices (denoted by
$\P_{st}$, where $s$ and $t$ are the prescribed endvertices), (simple)
circuits (denoted by $\C$: by that we mean a closed walk of length
at least 2, without edge- and node-repetition), forests ($\F$), spanning trees
($\spt$), (not necessarily spanning) trees ($\T$), and cuts (denoted
by $\cut$). 
Let $G=(V,E)$ be a \textbf{connected} undirected graph (we
assume connectedness in order to avoid trivial case-checkings) and
$\A$ and $\B$ two (not necessarily different) object types from the 7
possibilities above. The general questions we ask are the following:
\begin{itemize}
\item \textbf{Packing problem} (denoted by $\A \d \B$): can we \textbf{find two edge-disjoint subgraphs} in $G$, one of type $\A$ and the other of type $\B$?
\item \textbf{Covering problem} (denoted by $\A \c \B$): can we \textbf{cover the edge set} of $G$ with an object of type $\A$ and an object of type $\B$?
\item \textbf{Partitioning problem} (denoted by $\A \u \B$): can we \textbf{partition the edge set} of $G$ into an object of type $\A$ and an object of type $\B$?
\end{itemize}

Let us give one example of each type. 
A typical partitioning problem is the following: decide whether the
edge set of $G$ can be partitioned into a spanning tree and a
forest. Using our notations this is Problem $\spt \u \F$.  This problem
is in \textbf{NP $\cap$ co-NP} by the results of Nash-Williams \cite{nw},
polynomial algorithms for deciding the problem were given by Kishi and
Kajitani \cite{kishi3}, and Kameda and Toida \cite{kameda}.

A typical packing problem is the following: given four (not
necessarily distinct) vertices $s,t,s',t'\in V$, decide whether there
exists an $s$-$t$ path $P$ and an $s'$-$t'$-path $P'$ in $G$, such that
$P$ and $P'$ do not share any edge. With our notations this is Problem
$\P_{st} \d \P_{s't'} $. This 
problem is still
solvable in polynomial time, as was shown by Thomassen
\cite{thomassen} and Seymour \cite{seymour}.

A typical covering problem is the following: decide whether the edge
set of $G$ can be covered by a path and a 
circuit. In our notations this is Problem $P\c C$. Interestingly we found
that this simple-looking problem is NP-complete.

Let us introduce the following short formulation for the partitioning
and covering problems. If the edge set of a graph $G$ can be
partitioned into a type $A$ subgraph and a type $B$ subgraph then we
will also say that \textbf{\boldmath the edge set of $G$ is $A\u
  B$}. Similarly, if there is a solution of Problem $A\c B$ for a
graph $G$ then we say that \textbf{\boldmath the edge set of $G$ is
  $A\c B$}.

\begin{table}[!ht]
\begin{center}
  \caption{25 PARTITIONING PROBLEMS}
\label{tab:part}
\medskip

\begin{tabular}{|c|c|l|l|}
\hline
\textbf{Problem}&\textbf{Status}&\textbf{Reference}&\textbf{Remark}\\ \hline\hline
$\P \u \P$&\NPC&Theorem \ref{thm:part} & \NPC for subquadratic planar\\\hline
$\P \u \pst$&\NPC&Theorem \ref{thm:part}&  \NPC for subquadratic planar\\\hline
$\P \u \C$&\NPC&Theorem \ref{thm:part}&  \NPC for subquadratic planar\\ \hline
$\P \u \T$&\NPC&  Theorem \ref{thm:part} 
 & \NPC for subquadratic planar \\ \hline
$\P \u \spt$&\NPC & Theorem \ref{thm:part} 
&\NPC for subquadratic planar \\ \hline
$\P \u \F$&\NPC &Theorem \ref{thm:cut} (and Theorem \ref{thm:part}) &\NPC for subcubic planar \\ \hline
$\pst \u \P_{s't'}$&\NPC&Theorem \ref{thm:part}& \NPC for subquadratic planar\\\hline
$\pst \u \C$&\NPC&Theorem \ref{thm:part}& \NPC for subquadratic planar\\\hline
$\pst \u \T$&\NPC & Theorem \ref{thm:part} 
& \NPC for subquadratic planar\\ \hline
$\pst \u \spt$&\NPC& Theorem \ref{thm:part} 
&  \NPC for subquadratic planar\\ \hline
$\pst \u \F$&\NPC  &Theorem \ref{thm:cut}  (and Theorem \ref{thm:part})& \NPC for subcubic  planar\\ \hline
$\C \u \C$&\NPC&Theorem \ref{thm:part}& \NPC for subquadratic planar \\ \hline
$\C \u \T$&\NPC  & Theorem \ref{thm:part} 
&\NPC for subquadratic planar \\ \hline
$\C \u \spt$&\NPC & Theorem \ref{thm:part} 
& \NPC for subquadratic planar\\ \hline
$\C \u \F$&\NPC  &Theorem \ref{thm:cut}  (and Theorem \ref{thm:part}) & 
\NPC for  subcubic planar\\ \hline
$\T \u \T$&\NPC & P\'alv\"olgyi \cite{Dome} & planar graphs? \\ \hline
$\T \u \spt$&\NPC & Theorem \ref{thm:3} 
&  planar graphs? \\ \hline
$\F \u \F$&\inP & Kishi and Kajitani \cite{kishi3}, & in \inP for matroids:\\
 & &   Kameda and Toida \cite{kameda} & Edmonds  \cite{edmonds}\\
&&  (Nash-Williams \cite{nw}) & \\ \hline
$\spt \u \spt$&\inP &Kishi and Kajitani \cite{kishi3}, & in \inP for matroids:\\
 & &   Kameda and Toida \cite{kameda}, &  Edmonds  \cite{edmonds}\\
 & &  (Nash-Williams \cite{nw61}, &\\
 & &  Tutte \cite{tutte}) & \\ \hline
$\cut \u \cut$&\inP&  \iff bipartite &\\
& &(and $|V|\ge 3$) & \\ \hline
$\cut \u \F$&\NPC& Theorem \ref{thm:cut+F} & planar graps?  \\ \hline
$\cut \u \C$&\NPC&Theorem \ref{thm:cut}& \NPC for subcubic planar  \\ \hline
$\cut \u \T$&\NPC&Theorem \ref{thm:cut}& \NPC for subcubic planar  \\ \hline
$\cut \u \P$&\NPC&Theorem \ref{thm:cut}& \NPC for subcubic planar  \\ \hline
$\cut \u \pst$&\NPC&Theorem \ref{thm:cut}& \NPC for subcubic planar  \\ \hline
\end{tabular}
\end{center}
\end{table}

The setting outlined above gives us 84 problems. Note however that
some of these can be omitted. For example $\P\d \A$ is trivial for each
possible type $A$ in question, because $P$ may consist of only one vertex.  By the same reason, $\T\d \A$ and $\F\d \A$ type problems are also
trivial.  Furthermore, observe that the edge-set $E(G)$ of a graph $G$
is $\F\u \A $ $\Leftrightarrow$ $E(G)$ is $ \F\c \A$ $ \Leftrightarrow$
$E(G)$ is $ \T\c \A \Leftrightarrow$ $E(G)$ is $ \spt \c \A$:
therefore we will only consider the problems of form $\F\u \A$ among
these for any $\A$. Similarly, the edge set $E(G)$ is $\F\u \F $
$\Leftrightarrow $ $E(G)$ is $\T\u \F $ $\Leftrightarrow$ $E(G)$ is $
\spt \u \F$:
again we choose to deal with $\F\u \F$.
We can also omit the problems $\cut \u \spt$ and $\cut \d \spt$ because a
cut and a spanning tree can never be disjoint.

The careful calculation gives that we are left with 44 problems.  We
have investigated the status of these. Interestingly, many of these
problems turn out to be NP-complete. Our results are summarized in
Tables \ref{tab:part}-\ref{tab:cov}.  We
note that in our NP-completeness proofs we always show that the
considered problem is NP-complete even if the input graph is simple.
On the other hand, the polynomial algorithms given here always work
also for multigraphs (we allow parallel edges, but we forbid loops).
Some of the results shown in the tables were already proved in the
preliminary version \cite{quickpf} of this paper: namely we have
already shown the NP-completeness of
Problems $\P \u \T$, 
$\P \u \spt$, 
$\pst \u \T$, 
$\pst \u \spt$, 
$\C \u \T$,
$\C \u \spt$, 
$\T \u \spt$,
$\pst \d \spt$,
and $\C \d \spt$ there.

\begin{table}[!ht]
\begin{center}
  \caption{9 PACKING PROBLEMS}
\label{tab:pack}
\medskip

\begin{tabular}{|c|c|l|l|}
\hline
\textbf{Problem}&\textbf{Status}&\textbf{Reference}&\textbf{Remark}\\ \hline\hline
$\pst \d \P_{s't'}$&\inP &Seymour \cite{seymour}, Thomassen \cite{thomassen} & \\ \hline
$\pst \d \C$&\inP & see Section \ref{sec:alg1} & \\ \hline
$\pst \d \spt$&\NPC &Theorem \ref{thm:3} 
& planar graphs?\\ \hline
$\C \d \C$&\inP & Bodlaender \cite{bodl}  &  \NPC in linear matroids\\
&&(see also Section \ref{sec:alg2})& (Theorem \ref{thm:CdC}) \\ \hline
$\C \d \spt$&\NPC &Theorem \ref{thm:3} 
& polynomial in \\
&&& planar graphs, \cite{marcin}\\ \hline
$\spt \d \spt$&\inP &Imai \cite{imai}, (Nash-Williams \cite{nw61},&   in \inP for matroids: \\
& &  Tutte \cite{tutte}) &  Edmonds  \cite{edmonds}\\ \hline
$\cut \d \cut$&\inP&always, if $G$ has two   &   \NPC in linear matroids\\
& &non-adjacent vertices & (Corollary \ref{cor:CutdCut})\\ \hline
$\cut \d \pst$&\inP&always, except if the graph is an & \\
& & $s$-$t$ path (with multiple copies &  \\
& &  for some edges)&\\ \hline
$\cut \d \C$&\inP&always, except if the graph is  &  \NPC in linear matroids  \\
& & a tree, a circuit, or a bunch of  &   ($\Leftrightarrow$ the matroid
is not\\
&&parallel edges & uniform, Theorem \ref{thm:CutdC})\\ \hline
\end{tabular}
\end{center}
\end{table}

\begin{table}[!ht]
\begin{center}
  \caption{10 COVERING PROBLEMS}
\label{tab:cov}
\medskip

\begin{tabular}{|c|c|l|l|}
\hline
\textbf{Problem}&\textbf{Status}&\textbf{Reference}&\textbf{Remark}\\ \hline\hline
$\P \c \P$&\NPC&Theorem \ref{thm:part} & \NPC for subquadratic planar\\ \hline
$\P \c \pst$&\NPC&Theorem \ref{thm:part} & \NPC for subquadratic planar\\ \hline
$\P \c \C$&\NPC&Theorem \ref{thm:part} & \NPC for subquadratic planar\\ \hline
$\pst \c \P_{s't'}$&\NPC&Theorem \ref{thm:part} & \NPC for subquadratic planar\\ \hline
$\pst \c \C$&\NPC&Theorem \ref{thm:part}&\NPC for subquadratic planar \\ \hline
$C \c \C$&\NPC&Theorem \ref{thm:part} & \NPC for subquadratic planar\\ \hline
$\cut \c \cut$&\NPC& \iff 4-colourable & always in planar \\
 & &  & Appel et al.  \cite{4szin}, \cite{4szin2} \\ \hline
$\cut \c \C$&\NPC& Theorem \ref{thm:cut} &  \NPC for subcubic planar \\ \hline
$\cut \c \P$&\NPC& Theorem \ref{thm:cut}&   \NPC for subcubic  planar\\ \hline
$\cut \c \pst$&\NPC& Theorem \ref{thm:cut}&   \NPC for  subcubic planar\\ \hline
\end{tabular}
\end{center}
\end{table}


Problems $\P_{st}\u \spt$ and $\T\u \spt$ were posed in the open
problem portal called ``EGRES Open'' \cite{egresopen} of the Egerv\'ary
Research Group.  Most of the NP-complete problems remain NP-complete
for planar graphs, though we do not know yet the status of Problems $\T
\u \T$, $\T \u \spt$, $\cut \u \F$ $\pst \d \spt$, and $\C \d \spt$ for
planar graphs, as indicated in the table.

We point out to an interesting phenomenon: planar duality and the
NP-completeness of Problem $\C \u \C$ gives that deciding whether the
edge set of a planar graph is the disjoint union of two \emph{simple}
cuts is NP-complete (a \textbf{simple cut}, or \textbf{bond} of a
graph is an inclusionwise minimal cut). In contrast, the edge set of a
graph is $\cut \u \cut$ \iff\ the graph is bipartite on at least 3
nodes\footnote{It is easy to see that the edge set of a connected
  bipartite graph on at least 3 nodes is $\cut \u \cut$. On the other
  hand, the intersection of a cut and a circuit contains an even
  number of edges, therefore the edge set of a non-bipartite graph
  cannot be $\cut \u \cut$.}, that is $\cut \u \cut$ is polinomially
solvable even for non-planar graphs.

Some of the problems can be formulated as a problem in the graphic
matroid and therefore also have a natural matroidal generalization. 
For example the matroidal generalization of $\C\d \C$ is the following:
can we find two disjoint circuits in a matroid (given with an independence
oracle, say)?
Of course, such a matroidal question is only interesting here if it
can be solved for graphic matroids in polynomial time. Some of these
matroidal questions is known to be solvable (e.g., the matroidal
version of $\spt \u \spt$), and some of them was unknown (at least for
us): the best example being the (above mentioned) matroidal version of
$\C \d \C$.
In the table above we indicate these matroidal generalizations, too,
where the meaning of the problem is understandable. The matroidal
generalization of spanning trees, forests, circuits is
straightforward. We do not want to make sense of trees, paths, or
$s$-$t$-paths in matroids.  On the other hand, cuts deserve some
explanation. In matroid theory, a \textbf{cut} (also called
\textbf{bond} in the literature) of a matroid is defined as an
inclusionwise minimal subset of elements that intersects every
base. In the graphic matroid this  corresponds to a simple cut of
the graph defined above. 
So we will only consider
packing problems for cuts in matroids: for example the problem of type
$\A\d \cut$ in graphs is equivalent to the problem of packing $A$ and a
simple cut in the graph, therefore the matroidal generalization is
understandable.  We will discuss these matroidal generalizations in
Section \ref{sec:matroid}.

\section{NP-completeness proofs}\label{sect:npc}

\newcommand{\planarregNPC}{{\sc Planar3\-Reg\-Ham}}
\newcommand{\planarregNPCe}{{\sc Planar3\-Reg\-Ham}$-e$}

A graph $G=(V,E)$ is said to be \textbf{subcubic} if $d_G(v)\le 3$ for
every $v\in V$.  In many proofs below we will use Problem
\planarregNPC\ and Problem \planarregNPCe\ given below.

\begin{prob}[\planarregNPC]
  Given a $3$-regular planar graph $G=(V,E)$,
decide whether there is a \Hcs in $G$.
\end{prob}

\begin{prob}[\planarregNPCe]
  Given a $3$-regular planar graph $G=(V,E)$ and an edge $e\in E$,
decide whether there is a \Hcs in $G$ through edge $e$.
\end{prob}

It is well-known that Problems \planarregNPC\ and \planarregNPCe\ are
NP-complete (see Problem [GT37] in \cite{gj}).


\subsection{NP-completeness proofs for subcubic planar graphs
}\label{sec:cut+c}

\begin{figure}[!ht]
\begin{center}
\input{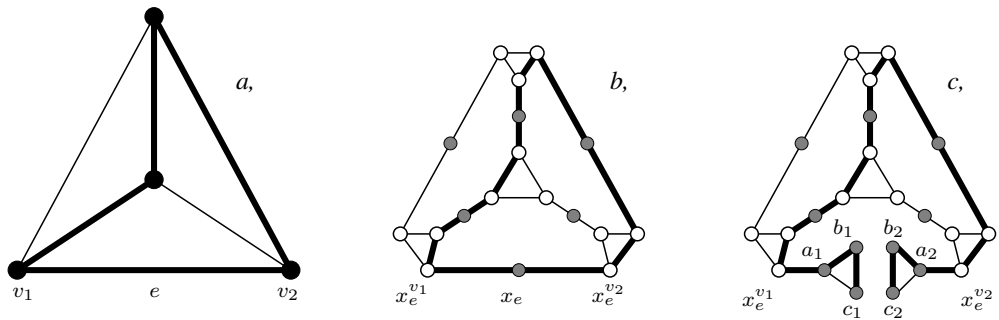}
\caption{An illustration for the proof of Theorem \ref{thm:cut}.}
\label{fig:k4min}
\end{center}
\end{figure}

\begin{thm}\label{thm:cut}
The following problems are NP-complete, even if restricted to subcubic
planar graphs: $\cut \c \C$, $\cut \u \C$, $\C\u \F$, $\cut \c \P$,
$\cut \c \pst$, $\cut \u \P$, $\cut \u \pst$, $\cut \u \T$, $\P\u \F$,
$\pst\u \F$.
\end{thm}

\begin{proof}
All the problems are clearly in NP.  First we prove the completeness
of $\cut \c \C$, $\cut \u \C$ and $\C \u \F$ using a reduction from Problem
\planarregNPC.
Given an instance of the Problem \planarregNPC\ with the 3-regular
planar graph $G$, construct the following graph $G'$.  First subdivide
each edge $e=v_1v_2\in E(G)$ with 3 new nodes $x_e^{v_1},x_e,x_e^{v_2}$ such that
they form a path in the order $v_1,x_e^{v_1},x_e,x_e^{v_2},v_2$. Now for any node
$u\in V(G)$ and any pair of edges $e,f\in E(G)$ incident to $u$
connect $x_e^u$ and $x_f^u$ with a new edge. Finally, delete all the
original nodes $v\in V(G)$ to get $G'$.
Informally speaking, $G'$ is obtained from $G$ by blowing a triangle
into every node of $G$ and subdividing each original edge with a new
node: see Figure \ref{fig:k4min} \textit{a,b,} for an illustration. Note that by
contracting these triangles in $G'$ and undoing the subdivision
vertices of form $x_e$ gives back $G$.
Clearly, the resulting graph $G'$ is still planar and has maximum
degree 3 (we mention that the subdivision nodes of form $x_e$ are only needed for
the Problem $\cut \u \C$).  
We will prove that $G$ contains a \Hcs
\iff $G'$ contains a circuit covering odd circuits (i.e., the edge-set of
$G'$ is $\C\c \cut$)
\iff the edge-set of $G'$ is $\C\u \cut$
\iff $G'$ contains a circuit covering all the circuits (i.e., the edge
set of $G'$ is $\C\u \F$).  First let $C$ be a Hamiltonian circuit in
$G$. We define a circuit $C'$ in $G'$ as follows. For any $v\in V(G)$,
if $C$ uses the edges $e, f$ incident to $v$ then let $C'$ use the 3
edges $x_ex_e^v, x_e^vx_f^v, x_f^vx_f$ (see Figure \ref{fig:k4min} \textit{a,b,} for an
illustration).  Observe that $G'-C'$ is a forest, proving that the
edge-set of $G'$ is $\C\u \F$. Similarly, the edge set of $G'-C'$ is a
cut of $G'$, proving that the edge-set of $G'$ is $\C\u \cut$. Finally
we show that if the edge set of $G'$ is $\C\c \F$ then $G$ contains a
Hamiltonian circuit: this proves the sequence of equivalences stated
above (the remaining implications being trivial). Assume that $G'$ has
a circuit $C'$ that intersects the edge set of every odd circuit of
$G'$. Contract the triangles of $G'$ and undo the subdivision nodes of
form $x_e$ and observe that $C'$ becomes a Hamiltonian circuit of $G$.

For the rest of the problems we use \planarregNPCe.  Given
the 3-regular planar graph $G$ and an edge
$e=v_1v_2\in E(G)$, first construct the graph $G'$ as above. Next modify $G'$
the following way: if $x_e^{v_1},x_e,x_e^{v_2}$ are the nodes of $G'$ arising from the subdivision of the original edge $e\in E(G)$ 
then let $G''=(G'-x_e)+\{x_e^{v_i}a_i, a_ib_i, b_ic_i, c_ia_i: i=1,2\}$,
where $a_i,b_i,c_i (i=1,2)$ are 6 new nodes (informally, ``cut off'' the path
$x_e^{v_1},x_e,x_e^{v_2}$ at $x_e$ and substitute
the arising two vertices of degree 1 with two triangles). 
An illustration can be seen in Figure \ref{fig:k4min} \textit{a,c}.

Let $s=c_1$ and $t=c_2$. 
The following chain of equivalences settles the NP-completeness of the
rest of the problems promised in the theorem. The proof is similar to
the one above and is left to the reader.


There exists a \Hcs in $G$ using the edge $e$ $\Leftrightarrow$ the
edge set of $G''$ is $\cut \u \pst$ $\Leftrightarrow$ the edge set of
$G''$ is $\cut \u \P$ $\Leftrightarrow$ the edge set of $G''$ is $\cut \u
T$ $\Leftrightarrow$ the edge set of $G''$ is $\cut \c \pst$
$\Leftrightarrow$ the edge set of $G''$ is $\cut \c \P$
$\Leftrightarrow$ the edge set of $G''$ is $\pst \u \F$
$\Leftrightarrow$ the edge set of $G''$ is $\P \u \F$.
\end{proof}

\subsection{NP-completeness proofs based on Kotzig's theorem}

Now we prove the NP-completeness of many other problems in our
collection using
the following elegant result
proved by Kotzig \cite{kotzig}.

\begin{thm}\label{thm:kotzig}
A 3-regular graph contains a Hamiltonian circuit if and only if the
edge set of its  line graph can be decomposed into two
Hamiltonian circuits.
\end{thm}

This theorem was used to prove NP-completeness results by Pike in 
\cite{pike}. 
Another useful and well known observation is the following: 
the line graph of a planar 3-regular graph is 4-regular and planar.




\begin{thm}\label{thm:part}
  The following problems are NP-complete, even if restricted to subquadratic planar graphs: $\P\u \P$, $\P\u \P_{st}$, $\P\u \C$, $\P\u \T$,
  $\P\u \spt$, $\P\u \F,$ $\P_{st}\u \P_{s't'}, \; \P_{st}\u \C$, $\P_{st}\u
  \F$, $\P_{st}\u \T$, $\P_{st}\u \spt$, $\C\u \C$, $\; \C\u \T$, $\; \C\u
  \spt$, $\C\u \F$, $\P\c \P$, $\P\c \P_{st}$, $\P\c \C$, $\P_{st}\c \P_{s't'}$, $\P_{st}\c \C$, 
$\C\c \C$.
\end{thm}
\begin{proof}
The problems are clearly in NP. Let $G$ be a planar 3-regular
graph. Since $L(G)$ is 4-regular, it is decomposable to two circuits
\iff it is decomposable to two \Hc s. This together with Kotzig's
theorem shows that $\C\u \C$ is NP-complete. For every other problem of
type $\C\u \A$ use $L=L(G)\!-\!st$ with an arbitrary edge $st$ of $L(G)$.
Let $C$ be a circuit of $L$ and observe that (by the number of edges
of $L$ and the degree conditions) $L\!-\!C$ is circuit-free if and only if
$C$ is a \Hcs and $L\!-\!C$ is a Hamiltonian path connecting $s$ and $t$.

For the rest of the partitioning type problems we need one more trick.  Let us be given
a $3$-regular planar graph $G=(V,E)$ and an edge $e=xy\in E$. We
construct another $3$-regular planar graph $G'=(V',E')$ as
follows. Delete edge $xy$, add vertices $x', y'$, and add edges $xx',
yy'$ and add two parallel edges between $x'$ and $y'$, namely $e_{xy}$
and $f_{xy}$ (note that $G'$ is planar, too).  Clearly $G$ has a \Hcs
through edge $xy$ if and only if $G'$ has a \Hc. Now consider $L(G')$,
the line graph of $G'$, it is a $4$-regular planar graph.
Note, that in $L(G')$ there are two parallel edges
between nodes $s=e_{xy}$ and $t=f_{xy}$, call these edges $g_1$ and
$g_2$. Clearly, $L(G')$ can be decomposed into two \Hc s \iff
$L'=L(G')\!-\!g_1\!-\!g_2$ can be decomposed into two \Hp s. Let $P$ be a path
in $L'$ and notice again that the number of edges of $L'$ and the
degrees of the nodes in $L'$ imply that $L'\!-\!P$ is circuit free
\iff\ $P$ and $L'\!-\!P$ are two \Hp s in $L'$.

Finally, the NP-completeness of the problems of type $\A\c \B$ is an
easy consequence of the NP-completeness of the corresponding
partitioning problem $\A\u \B$: use the same construction and observe
that the number of edges enforce the two objects in the cover to be
disjoint.
\end{proof}

We remark that the above theorem gives a new proof of the
NP-completeness of Problems $\C\u \F$, $\P\u \F$ and $\pst\u \F$, already
proved in Theorem \ref{thm:cut}.

\subsection{NP-completeness of Problems $\P_{st}\d \spt$,   $\T\u \spt$, $\C \d \spt$, and $\cut \u \F$}

First we show the NP-completeness of Problems $\P_{st}\d \spt$, $\T\u
\spt$, and $\C \d \spt$.  Problem $\T\u \T$ was proved to be NP-complete
by P\'alv\"olgyi in \cite{Dome}
(the NP-completeness of this problem with the additional requirement that
the two trees have to be of equal size was proved by Pferschy,
Woeginger and Yao \cite{woeg}). Our reductions here are similar to the
one used by P\'alv\"olgyi in \cite{Dome}. We remark that our first
proof for the NP-completeness of Problems $\P \u \T$, $\P \u \spt$, $\pst
\u \T$, $\pst \u \spt$, $\C \u \T$ and $\C \u \spt$ used a variant of the
construction below (this can be found in \cite{quickpf}), but later we
found that using Kotzig's result (Theorem \ref{thm:kotzig}) a simpler
proof can be given for these.

For a subset of edges $E'\subseteq E$ in a graph $G=(V,E)$, let
$V(E')$ denote the subset of nodes incident to the edges of $E'$,
i.e.,  $V(E')=\{v\in V: $ there exists an $f\in E' $ with $v\in f\}$.

\begin{thm}\label{thm:3}
Problems $\P_{st}\d \spt$, $\T \u \spt$ and $\C \d \spt$ are NP-complete
even for graphs with maximum degree 3.
\end{thm}

\begin{proof}
It is clear that the problems are in NP.  Their completeness will be
shown by a reduction from the well known NP-complete problems
\threesat\ or the problem \oneinthree\ (Problems LO2 and LO4 in
\cite{gj}).  Let $\varphi$ be a 3-CNF formula with variable set
$\{x_1,x_2,\dots,x_n\}$ and clause set $\cC=\{C_1,C_2,\dots,C_m\}$
(where each clause contains exactly 3 literals).  Assume that literal
$x_j$ appears in $k_j$ clauses
$C_{a^j_{1}},C_{a^j_{2}},\dots,C_{a^j_{k_j}}$, and literal $\ov{x_j}$
occurs in $l_j$ clauses $C_{b^j_{1}},C_{b^j_{2}},\dots,C_{b^j_{l_j}}$.
Construct the following graph $G_\varphi=(V,E)$. 

For an arbitrary clause $C\in \cC$ we will introduce a new node $u_C$,
and for every literal $y$ in $C$ we introduce two more nodes $v(y,C), w(y,C)$.
Introduce the edges $u_Cw(y,C), w(y,C)v(y,C)$
for every clause $C$ and every literal $y$ in $C$ 
(the nodes $w(y,C)$ will have degree 2).

For every variable $x_j$ introduce 8 new nodes $z^j_1, z^j_2,$ $ w^j_1,
\ov{w^j_1},$ $w^j_2,$ $ \ov{w^j_2},$ $w^j_3, \ov{w^j_3}$.  For every
variable $x_j$, let $G_\varphi$ contain a circuit on the $k_j+l_j+4$
nodes $z^{j}_1,$ $ v(x_j,C_{a^j_{1}}),$ $v(x_j,C_{a^j_{2}}),$ $\dots,$ $
v(x_j,C_{a^j_{k_j}}),$ $ w^j_1, $ $z^j_2,$ $ \ov{w^j_1},$ $
v(\ov{x_j},C_{b^j_{l_j}}),$ $v(\ov{x_j},C_{b^j_{l_j-1}}), $ $\dots,$ $
v(\ov{x_j},C_{b^j_{1}})$ in this order.
We say that this circuit is \textbf{ associated to variable
  $x_j$}. Connect the nodes $z^j_2$ and $z^{j+1}_1$ with an edge for
every $j=1,2,\dots,n\!-\!1$.  Introduce furthermore a path on nodes
$w^1_3,\ov{w^1_3},w^2_3,\ov{w^2_3}, \dots, w^n_3,\ov{w^n_3}$ in this
order and add the edges $w^j_1w^j_2, w^j_2w^j_3, \ov{w^j_1}\ov{w^j_2},
\ov{w^j_2}\ov{w^j_3}$ for every $j=1,2,\dots, n$.  Let $s=z^1_1$ and
$t=z^n_2$.
 
\begin{figure}[!ht]
\begin{center}
\input{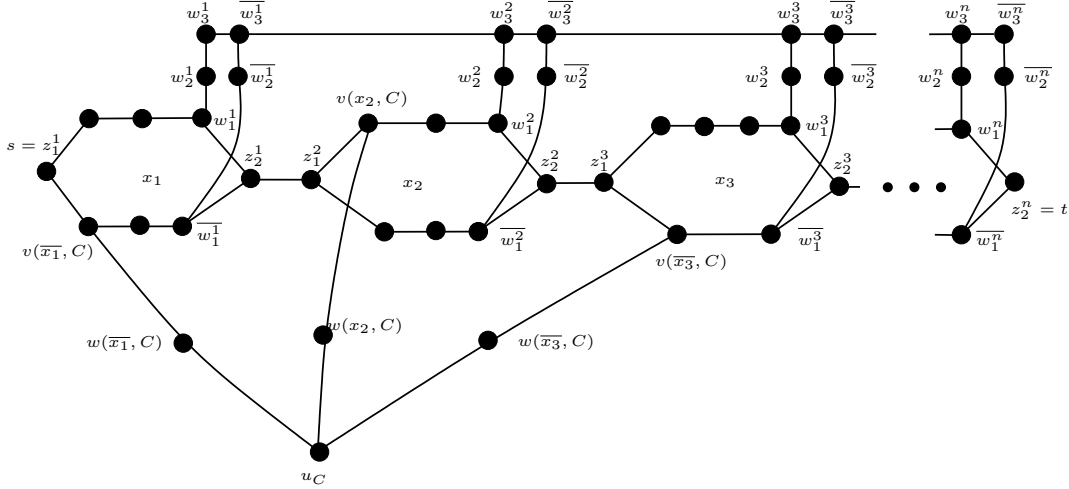}
\caption{Part of the construction of graph $G_\varphi$ for clause
  $C=(\ov{x_1}\vee {x_2}\vee \ov{x_3})$.}\label{fig:G_phi}
\end{center}
\end{figure}

The construction of the graph $G_\varphi$ is
finished.  An illustration can be found in Figure \ref{fig:G_phi}.  

Clearly, $G_\varphi$ is simple and has maximum degree three. 

If $\tau$ is a truth assignment to the variables $x_1,x_2,\dots,x_n$
then we define an \stp  $P_\tau$ as follows: for every
$j=1,2,\dots,n$, if $x_j$ is set to TRUE then let $P_\tau$ go through
the nodes $z^j_1, v(\ov{x_j},C_{b^j_{1}}),
v(\ov{x_j},$ $C_{b^j_{2}}),\dots ,$ $v(\ov{x_j},$ $C_{b^j_{l_j}}),$ $\ov{w^j_1},z^j_2$,
otherwise (i.e., if $x_j$ is set to FALSE) let $P_{\tau}$ go through $z^j_1,
v({x_j},$ $C_{a^j_{1}}), $ $v({x_j},$ $C_{a^j_{2}}),\dots
,$ $v({x_j},$ $C_{a^j_{k_j}}),{w^j_1}, z^j_2$.

We need one more concept. An \stp $P$ is called an
\emph{assignment-defining path} if 
 $v\in V(P),\ d_G(v)=2$ implies $v\in \{s,t\}$.
For such a path $P$ we define the truth assignment $\tau_P$ such that
$P_{\tau_P}=P$. 



\begin{cl}
There is an \stp  $P\subseteq E$ such that $(V,\; E\!-\!P)$ is
connected if and only if $\varphi\in 3SAT$. Consequently, 
Problem $\pst \d \spt$ is NP-complete.
\end{cl}

\begin{proof}
If $\tau$ is a truth assignment showing that $\varphi\in 3SAT$ then
$P_\tau$ is a path satisfying the requirements, as one can check. On
the other hand, if $P$ is an \stp such that $(V,\; E\!-\!P)$ is connected
then $P$ cannot go through nodes of degree 2, therefore $P$ is
assignment-defining,
and $\tau_P$ shows 
$\varphi\in 3SAT$.
\end{proof}

To show the NP-completeness of Problem
$\T \u \spt$, modify $G_{\varphi}$ the following way: subdivide the two
edges incident to $s$ with two new nodes $s'$ and $s''$ and connect
these two nodes with an edge. Repeat this with $t$: subdivide the two
edges incident to $t$ with two new nodes $t'$ and $t''$ and connect
$t'$ and $t''$. Let the graph obtained this way be $G=(V,E)$.
Clearly, $G$ is subcubic and simple. Note that the definition of an
assignment defining path and that of $P_\tau$ for a truth assignment
$\tau $ can be obviously modified for the graph $G$.

\begin{cl}
There exists a truth assignment $\tau$ such that every clause in $\varphi$
contains exactly one true literal if and only if there exists a set
$T\subseteq E$ such that $(V(T),\, T)$ is a tree and $(V,\; E\!-\!T)$ is a
spanning tree. Consequently, 
Problem $\T \u \spt$ is NP-complete.
\end{cl}
\begin{proof}
If $\tau$ is a truth assignment as above then one can see that
$T=P_\tau$ is an edge set satisfying the requirements.

On the other hand, assume that $T\subseteq E$ is such that $(V(T),\,
T)$ is a tree and $T^*=(V,\; E\!-\!T)$ is a spanning tree. Since $T^*$
cannot contain circuits, $T$ must contain at least one of the 3 edges
$ss',s's'',s''s$ (call it $e$), as well as at least one of the 3 edges
$tt',t't'',t''t$ (say $f$). Since $(V(T),\, T)$ is connected, $T$
contains a path $P\subseteq T$ connecting $e$ and $f$ (note that since
$(V, E\!-\!T)$ is connected, $|T\cap \{ss',s's'',s''s\}|=|T\cap
\{tt',t't'',t''t\}|=1$). 
Since $(V,\; E\!-\!P)$ is connected, $P$ cannot go through nodes of
degree 2 (except for the endnodes of $P$), and the edges $e$ and $f$
must be the last edges of $P$ (otherwise $P$ would disconnect $s$ or
$t$ from the rest). Thus, without loss of generality we can assume
that $P$ connects $s$ and $t$ (by locally changing $P$ at its ends),
and we get that $P$ is assignment defining. Observe that in fact $T$
must be equal to $P$, since $G$ is subcubic (therefore $T$ cannot
contain nodes of degree 3).  Consider the truth assignment $\tau_P$
associated to $P$, we claim that $\tau_P$ satisfies our requirements.
Clearly, if a clause $C$ of $\varphi$ does not contain a true literal
then $u_C$ is not reachable from $s$ in $G\!-\!T$, therefore every clause
of $\varphi$ contains at least one true literal. On the other hand
assume that a clause $C$ contains at least 2 true literals (say $x_j$
and $\ov{x_k}$ for some $j\ne k$), then one can see that there exists
a circuit in $G\!-\!T$ (because $v(x_j,C)$ is still reachable from
$v(\ov{x_k},C)$ in $G\!-\!T\!-\!u_C$ via the nodes $w_j^1, w_j^2,w_j^3$ and
$\ov{w_k^1}, \ov{w_k^2},\ov{w_k^3}$).
\end{proof}

Finally we prove the NP-completeness of Problem $\C \d \spt$.  For the
3CNF formula $\varphi$ with variables $x_1,x_2,\dots,x_n$ and clauses
$C_1,C_2,\dots,C_m$, let us associate the 3CNF formula $\varphi'$ with
the same variable set and clauses $(x_1 \vee x_1 \vee \ov{x_1}),\;
(x_2 \vee x_2 \vee \ov{x_2}), \dots,\; (x_n \vee x_n \vee \ov{x_n}),\; 
C_1,\; C_2,\dots, C_m$. Clearly, $\varphi$ is satisfiable \iff\ $\varphi'$
is satisfiable.
Construct the graph $G_{\varphi'}=(V,E)$ as
above (the construction is clear even if some clauses contain only 2
literals), and let $G=(V,E)$ be obtained from $G_{\varphi'}$ by adding
the edge $st$.

\begin{cl}
The formula $\varphi'$ is satisfiable if and only if there exists a
set $K\subseteq E$ such that $(V(K),\, K)$ is a circuit and $G\!-\!K=(V,\;
E\!-\!K)$ is connected. Consequently, Problem $\C \d \spt$ is NP-complete.
\end{cl}
\begin{proof}
First observe that if $\tau$ is a truth assignment satisfying
$\varphi'$ then $K=P_\tau\cup\{st\}$ is an edge set satisfying the
requirements.  On the other hand, if $K$ is an edge set satisfying the
requirements then $K$ cannot contain nodes of degree 2, since $G\!-\!K$ is
connected. We claim that $K$ can neither be a circuit associated to a
variable $x_i$, because in this case the node $u_C$ associated to
clause $C=(x_i \vee x_i \vee \ov{x_i})$ would not be reachable in
$G\!-\!K$ from $s$. Therefore $K$ consists of the edge $st$ and an
assignment defining path $P$. It is easy to check (analogously to the
previous arguments) that $\tau_P$ is a truth assignment satisfying
$\phi'$.
\end{proof}

\noindent As we have proved the NP-completeness of all three problems,
the theorem is proved.
\end{proof}


We note that the construction given in our original proof of the above
theorem (see \cite{quickpf})
was used  by Bang-Jensen and Yeo in \cite{bjyeo}.  They
settled an open problem raised by Thomass\'e in 2005.  They proved
that it is NP-complete to decide $\mathrm{SpA}\wedge \spt$ in
digraphs, where $\mathrm{SpA}$ denotes a spanning arborescence and
$\spt$ denotes a spanning tree in the underlying undirected graph.

We also point out that the planarity of the graphs in the above proofs
cannot be assumed. We do not know the status of any of the Problems
$\P_{st}\d \spt$, $\T\u \spt$, and $\T\u \T$ in planar graphs.  It was
shown in \cite{marcin} that Problem $\C \d \spt$ is polynomially
solvable in planar graphs.  We also mention that planar duality gives
that Problem $\C\d \spt$ in a planar graph is eqivalent to finding a
cut in a planar graph that contains no circuit: by the results of
\cite{marcin}, this problem is also polynomially solvable.  However
van den Heuvel \cite{heuvel} has shown that this problem is 
NP-complete for general (i.e., not necessarily planar) graphs.

We point out to an interesting connection towards the Graphic TSP
Problem. This problem can be formulated as follows. Given a connected
graph $G=(V,E)$, find a connected Eulerian subgraph of $2G$ spanning
$V$ with minimum number of edges (where $2G=(V,2E)$ is the graph
obtained from $G$ by doubling its edges). The connection is the
following. Assume that $F\subseteq 2E$ is a feasible solution to the
problem. A greedy way of improving $F$ would be to delete edges from it, while
maintaining the feasibility.  It is thus easy to observe that this
greedy improvement is possible if and only if the graph $(V,F)$
contains an edge-disjoint circuit and a spanning tree (which is
Problem $\C \d \spt$ in our notations). However, slightly modifying
the proof above it can be shown that Problem $\C \d \spt$ is also
NP-complete in Eulerian graphs (details can be found in
\cite{marcin}).



\begin{thm}\label{thm:cut+F}
Problem $\cut \u \F$ is NP-complete.
\end{thm}
\begin{proof}
The problem is clearly in NP.  In order to show its completeness let
us first rephrase the problem. Given a graph, Problem $\cut \u \F$ asks
whether we can colour the nodes of this graph with two colours such
that no monochromatic circuit exists.

Consider the  NP-complete Problem 
2-COLOURABILITY OF A 3-UNIFORM HYPERGRAPH 
This problem is the following: given a 3-uniform hypergraph $H=(V,
\cE)$, can we colour the set $V$ with two colours (say red and blue)
such that there is no monochromatic hyperedge in $\cE$ (the problem is
indeed NP-complete, since Problem GT6 in \cite{gj} is a special case
of this problem).  Given the instance $H=(V,\cE)$ of this problem,
construct the following graph $G$. The node set
of $G$ is $V\cup V_{\cE}$, where $V_{\cE}$ is disjoint from $V$ and it
contains 6 nodes for every hyperedge in $\cE$: for an arbitrary
hyperedge $e=\{v_1,v_2,v_3\}\in \cE$, the 6 new nodes associated to it
are $x_{v_1,e}, y_{v_1,e}, x_{v_2,e}, y_{v_2,e}, x_{v_3,e}, y_{v_3,e}$. 
The edge set of $G$ contains the following edges: for the hyperedge
$e=\{v_1,v_2,v_3\}\in \cE$, $v_i$ is connected with  $x_{v_i,e}$ and
$y_{v_i,e}$ for every $i=1,2,3$, and among the 6 nodes associated to $e$
every two is connected with an edge
except for the 3 pairs of form $x_{v_i,e},y_{v_i,e}$ for $i=1,2,3$
(i.e., $|E(G)|=18|\cE|$).  
The construction of $G$ is finished. An illustration can be found in
Figure \ref{fig:cutplF}. Note that in any two-colouring of $V\cup
V_{\cE}$ the 6 nodes associated to the hyperedge $e=\{v_1,v_2,v_3\}\in
\cE$ do not induce a monochromatic circuit if and only if there exists
a permutation $i,j,k$ of $1,2,3$ so that they are coloured the
following way: $x_{v_i,e},y_{v_i,e}$ is blue, $x_{v_j,e},y_{v_j,e}$ is
red and $x_{v_k,e},y_{v_k,e}$ are of different colour.

\begin{figure}[!ht]
\begin{center}
\input{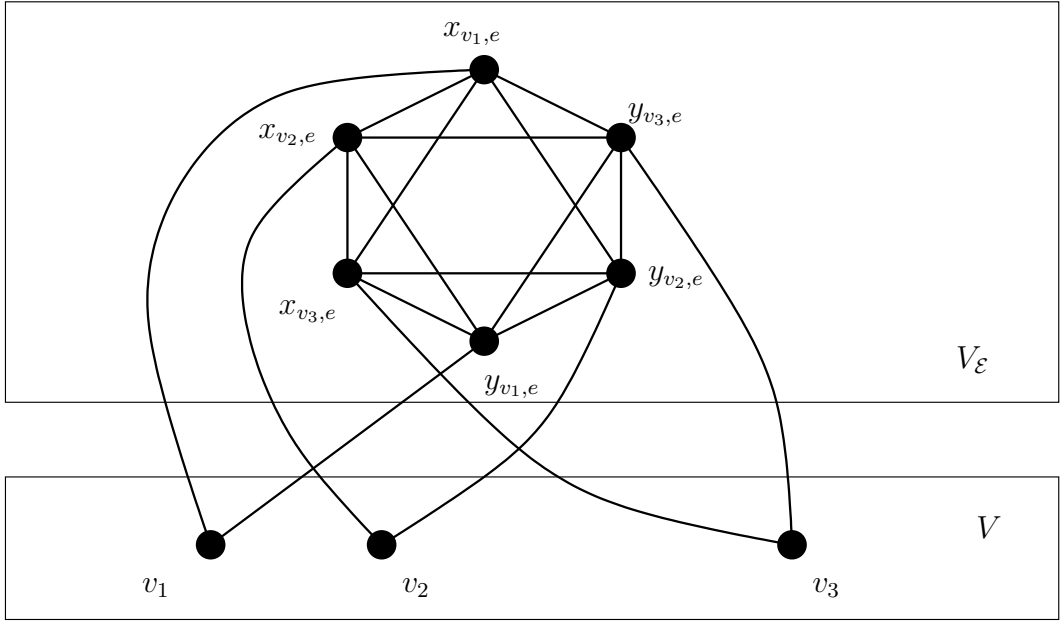}
\caption{Part of the construction of the graph $G$ in the proof of
  Theorem \ref{thm:cut+F}.}\label{fig:cutplF}
\end{center}
\end{figure}

One can check that $V$ can be coloured with 2 colours such that there
is no monochromatic hyperedge in $\cE$ \iff\ $V\cup V_{\cE}$ can be
coloured with 2 colours such that there is no monochromatic circuit in
$G$. 
\end{proof}

We note  that we do not know the status of Problem $\cut \u
\F$ in planar graphs.

\section{Algorithms} \label{sec:alg1}\label{sec:alg2}


\paragraph{Algorithm for $\pst \d \C$.}
  Assume we are given a connected multigraph $G=(V,E)$ and two nodes
  $s,t\in V$, and we want to decide whether an $s$-$t$-path
  $P\subseteq E$ and a circuit $C\subseteq E$ exists with $P\cap
  C=\emptyset$. We may even assume that both $s$ and $t$ have degree
  at least two.  If $v\in V\!-\!\{s,t\}$ has degree at most two then we
  can eliminate it.  If there is a cut-vertex $v\in V$ then we can
  decompose the problem into smaller subproblems by checking whether
  $s$ and $t$ fall in the same component of $G\!-\!v$, or not. If they do
  then $P$ should lie in that component, otherwise $P$ has to go
  through $v$.

If there is a non-trivial two-edge $s$-$t$-cut (i.e., a set $X$ with
$\{s\}\subsetneq X\subsetneq V\!-\!t$, and $d_G(X)=2$), then we can again
reduce the problem in a similar way:
the circuit to be found cannot use both edges entering $X$ and we have
to solve two smaller problems obtained by contracting $X$ for the
first one, and contracting $V\!-\!X$ for the second one. 

%
So we can assume that $|E|\ge n+\lceil n/2
\rceil-1$, and that $G$ is $2$-connected and $G$
has no non-trivial two-edge $s$-$t$-cuts.  
Run a BFS
from $s$ and associate levels to vertices ($s$ gets $0$).  If $t$ has
level at most $\lceil n/2 \rceil -1$ then we have a path of length at
most $\lceil n/2 \rceil -1$ from $s$ to $t$, after deleting its edges,
at least $n$ edges remain, so we are left with a circuit.

So we may assume that the level of $t$ is at least $\lceil n/2
\rceil$.  As $G$ is $2$-connected, we must have at least two vertices
on each intermediate level. Consequently $n$ is even, $t$ is on level
$n/2$, and we have exactly two vertices on each intermediate level,
and each vertex $v\in V\!-\!\{s,t\}$ has degree $3$, or, otherwise for a
minimum $s$-$t$ path $P$ we have that $G\!-\!P$ has at least $n$ edges,
i.e., it contains a circuit.  We have no non-trivial two-edge $s$-$t$-cuts,
consequently there can only be two cases: either $G$ equals to $K_4$ with
edge $st$ deleted, or $G$ arises from a $K_4$ such that two opposite
edges are subdivided (and these subdivision nodes are $s$ and $t$). In
either cases we have no solution.

\medskip

\paragraph{Algorithm for $\C \d \C$.}
We give a simple polynomial time algorithm for deciding whether two
edge-disjoint circuits can be found in a given connected multigraph
$G=(V,E)$. We note that a polynomial (but less elegant) algorithm for
this problem was also given in \cite{bodl}.

If any vertex has degree at most two, we can eliminate it, so we
may assume that the minimum degree is at least $3$.  If $G$ has at
least $16$ vertices, then it has a circuit of length at most $n/2$
(simply run a BFS from any node and observe that there must be a
non-tree edge between some nodes of depth at most $\log(n)$, giving us
a circuit of length at most $2\log(n)\le n/2$), and after deleting the
edges of this circuit, at least $n$ edges remain, so we are left with
another circuit. For smaller graphs we can check the problem in constant
time.



\section{Matroidal generalizations}\label{sec:matroid}

In this section we will consider the matroidal generalizations for the
problems that were shown to be polynomially solvable in the graphic
matroid. In fact we will only need linear matroids, since it turns out
that the problems we consider are already NP-complete in them.  We
will use the following result of Khachyan.

\begin{thm}[Khachyan \cite{khachyan}]\label{thm:khac}
Given a $D\times N$ matrix over the rationals, it is NP-complete to
decide whether there exist $D$ linearly dependent columns.
\end{thm}

First we consider the matroidal generalization of Problem $\cut \d \C$.

\begin{thm}\label{thm:CutdC}
It is NP-complete to decide whether an (explicitly given) linear matroid
contains a  cut and a circuit that are disjoint.
\end{thm}

\begin{proof}
Observe that there is no disjoint cut and circuit in a matroid if and
only if every circuit contains a base, that is equivalent with the
matroid being uniform.  
Khachyan's Theorem \ref{thm:khac} is equivalent with the uniformness
of the linear matroid determined by the coloumns of the matrix in
question, proving our theorem.
\end{proof}


Finally we consider the matroidal generalization of Problem $\C \d \C$
and $\cut \d \cut$.

\begin{thm}\label{thm:CdC}
The problem of deciding whether an (explicitly given) linear matroid
contains two disjoint circuits is NP-complete.
\end{thm}

\begin{proof}
We will prove here that Khachyan's Theorem \ref{thm:khac} is
true even if $N=2D+1$, which implies our theorem, since there are two
disjoint circuits in the linear matroid represented by this $D\times
(2D+1)$ matrix if and only if there are $D$ linearly dependent columns
in it.

Khachyan's proof of Theorem \ref{thm:khac} was simplified by Vardy
\cite{vardy}, we will follow his line of proof. Consider the following
problem.
\begin{prob}\label{prob:sum}
Given different positive integers $a_1,a_2,\dots,a_n,b$ and a positive
integer $d$, decide whether there exist $d$ indices $1\le
i_1<i_2<\dots<i_d\le n$ such that $b=a_{i_1}+a_{i_2}+\dots+a_{i_d}$.
\end{prob}

\newcommand{\subsum}{{\sc Subset-Sum}}

Note that Problem \ref{prob:sum} is very similar to the
\subsum\ Problem (Problem SP13 in \cite{gj}), the only difference
being that in the \subsum\ problem we do not specify $d$, and the
numbers $a_1,a_2,\dots,a_n$ need not be different. On the other hand,
here we will strongly need that the numbers $a_1,a_2,\dots,a_n$ are
all different.  Vardy has shown the following claim (we include a
proof for sake of completeness).

\begin{cl}\label{cl:vardy}
There is solution to Problem \ref{prob:sum} if and only if there are
$d+1$ linearly dependent columns (above the rationals) in the
$(d+1)\times (n+1)$ matrix
\[
\begin{pmatrix}
  1    & 1    & \cdots & 1 & 0 \\
  a_{1} & a_{2} & \cdots & a_{n}& 0 \\
  \vdots  & \vdots  & \ddots & \vdots  \\
  a_{1}^{d-2} & a_{2}^{d-2} & \cdots & a_{n}^{d-2} & 0 \\
  a_{1}^{d-1} & a_{2}^{d-1} & \cdots & a_{n} ^{d-1} & 1 \\
  a_{1}^{d} & a_{2}^{d} & \cdots & a_{n} ^{d} & b 
\end{pmatrix}.\] 
\end{cl}

\begin{proof}
We use the following facts about determinants. Given real numbers
$x_1,x_2,\dots,x_k$, 
we have the following well-known relation for the Vandermonde
determinant:
\[
\det\begin{pmatrix}
  1    & 1    & \cdots & 1  \\
  x_{1} & x_{2} & \cdots & x_{k}\\
  \vdots  & \vdots  & \ddots & \vdots  \\
  x_{1}^{k-1} & x_{2}^{k-1} & \cdots & x_{k} ^{k-1} 
\end{pmatrix}=\prod_{i<j}(x_j-x_i).
\]
Therefore the Vandermonde determinant is not zero, if the numbers
$x_1,x_2,\dots,x_k$ are different. Furthermore, we have the following
relation for an alternant of the Vandermonde determinant (see
Chapter V in \cite{muir}, for example):
\[
\det\begin{pmatrix}
  1    & 1    & \cdots & 1  \\
  x_{1} & x_{2} & \cdots & x_{k}\\
  \vdots  & \vdots  & \ddots & \vdots  \\
  x_{1}^{k-2} & x_{2}^{k-2} & \cdots & x_{k} ^{k-2} \\
  x_{1}^{k} & x_{2}^{k} & \cdots & x_{k} ^{k} 
\end{pmatrix}=(x_1+x_2+\dots+x_k)\prod_{i<j}(x_j-x_i).
\]
W include a proof of this last fact: given an arbitrary $k\times k$
matrix $X=((x_{ij}))$ and numbers $u_1,\dots,u_k$, observe (by checking the coefficients of the $u_i$s on each side) that 
\begin{eqnarray*}
\det\begin{pmatrix}
  u_1 x_{11} & u_2x_{12} & \cdots & u_kx_{1k}\\
  x_{21} & x_{22} & \cdots & x_{2k}\\
  \vdots  & \vdots  & \ddots & \vdots  \\
  x_{k1} & x_{k2} & \cdots & x_{kk} \\
\end{pmatrix}+
\det\begin{pmatrix}
  x_{11} & x_{12} & \cdots & x_{1k}\\
  u_1 x_{21} & u_2 x_{22} & \cdots & u_k x_{2k}\\
  \vdots  & \vdots  & \ddots & \vdots  \\
  x_{k1} & x_{k2} & \cdots & x_{kk} \\
\end{pmatrix}
+\dots +\\
\det\begin{pmatrix}
  x_{11} & x_{12} & \cdots & x_{1k}\\
  x_{21} & x_{22} & \cdots & x_{2k}\\
  \vdots  & \vdots  & \ddots & \vdots  \\
  u_1 x_{k1} & u_2x_{k2} & \cdots & u_k x_{kk} \\
\end{pmatrix}
=
(u_1+u_2+\dots u_k)\det\begin{pmatrix}
  x_{11} & x_{12} & \cdots & x_{1k}\\
  x_{21} & x_{22} & \cdots & x_{2k}\\
  \vdots  & \vdots  & \ddots & \vdots  \\
  x_{k1} & x_{k2} & \cdots & x_{kk} \\
\end{pmatrix}
\end{eqnarray*}
Now apply this to the Vandermonde matrix  $X=\begin{pmatrix}
  1    & 1    & \cdots & 1  \\
  x_{1} & x_{2} & \cdots & x_{k}\\
  \vdots  & \vdots  & \ddots & \vdots  \\
  x_{1}^{k-1} & x_{2}^{k-1} & \cdots & x_{k} ^{k-1} 
\end{pmatrix}$ and numbers $u_i=x_i$  for every $i=1,2,\dots,k$.
 
We will use these two facts. The first one implies that if $d+1$
columns of our matrix are dependent then they have to include the last
column. By the second fact, if $1\le
i_1<i_2<\dots<i_d\le n$ are arbitrary indices then 
\[
\det\begin{pmatrix}
  1    & 1    & \cdots & 1 & 0 \\
  a_{i_1} & a_{i_2} & \cdots & a_{i_d}& 0 \\
  \vdots  & \vdots  & \ddots & \vdots  \\
  a_{i_1}^{d-2} & a_{i_2}^{d-2} & \cdots & a_{i_d}^{d-2} & 0 \\
  a_{i_1}^{d-1} & a_{i_2}^{d-1} & \cdots & a_{i_d} ^{d-1} & 1 \\
  a_{i_1}^{d} & a_{i_2}^{d} & \cdots & a_{i_d} ^{d} & b 
\end{pmatrix}=(b-a_{i_1}-a_{i_2}-\dots-a_{i_d})\prod_{k<l}(a_{i_l}-a_{i_k}).\] 
This implies the claim.
\end{proof}

Vardy also claimed that Problem \ref{prob:sum} is NP-complete: our
proof will be completed if we show that this is indeed the case even
if $n=2d+2$. Since we have not found a formal proof of this claim of
Vardy, we will give a full proof of the following claim. For a set $V$
let ${V \choose 3}=\{X\subseteq V: |X|=3\}$.

\begin{cl}\label{cl:2dpl2}
Problem  \ref{prob:sum} is NP-complete even if $n=2d+2$. 
\end{cl}
\begin{proof}
\newcommand{\threeDM}{{\sc Exact-Cover-by-3-Sets}} We will reduce the
well-known NP-complete problem \threeDM\ (Problem SP2 in \cite{gj}) to
this problem. Problem \threeDM\ is the following: given a 3-uniform
family $\cE\subseteq {V \choose 3}$, decide whether
there exists a subfamily $\cE'\subseteq \cE$ such that every
element of $V$ is contained in exactly one member of $\cE'$. We assume
that 3 divides $|V|$, and let $d=|V|/3$, so Problem \threeDM\ asks
whether there exist $d$ disjoint members in $\cE$. First we show
that this problem remains NP-complete even if $|\cE|=2d+2$. Indeed, if
$|\cE|\ne 2d+2$ then let us introduce $3k$ new nodes $\{u_i,v_i,w_i:
i=1,2,\dots,k\}$ where
\begin{itemize}
\item $k$ is such that ${3k \choose 3}-2k\ge 2d+2-|\cE|$ if $|\cE| < 2d+2$, and 
\item $k=|\cE|-(2d+2)$, if $|\cE| > 2d+2$.
\end{itemize}
Let $V^*=V\cup\{u_i,v_i,w_i: i=1,2,\dots,k\}$ and let $\cE^*=\cE\cup
\{\{u_i,v_i,w_i\}: i=1,2,\dots,k\}$ (note that $ |V^*|=3(d+k)$). If
$|\cE| < 2d+2$ then include furthermore $ 2(d+k) + 2 - (|\cE|+k)$ arbitrary new
 sets of size 3 to $\cE^*$ from $ {V^*-V \choose 3}$, but so that $\cE^*$
does not contain a set twice (this can be done by the choice of
$k$).  It is easy to see that $|\cE^*|= 2|V^*|/3+2$, and $V$ can be
covered by disjoint members of $\cE$ if and only if $V^*$ can be
covered by disjoint members of $\cE^*$.

Finally we show that \threeDM\ is a special case of Problem
\ref{prob:sum} in disguise. Given an instance of \threeDM\ by a
3-uniform family $\cE\subseteq {V \choose 3}$, consider the
characteristic vectors of these 3-sets as different positive
integers written in base 2 (that is, assume a fixed ordering of the
set $V$, then the characteristic vectors of the members in $\cE$
are 0-1 vectors corresponding to different binary integers, conatining
3 ones in their representation). These will be the numbers
$a_1,a_2,\dots,a_{|\cE|}$. Let $b=2^{|V|}-1$ be the number
corresponding to the all 1 characteristic vector, and let
$d=|V|/3$. Observe that there exist $d$ disjoint members in $\cE$
if and only if there are indices $1\le i_1<i_2<\dots<i_d\le |\cE|$
such that $b=a_{i_1}+a_{i_2}+\dots+a_{i_d}$.
(You will need to prove a small claim about the maximal number of ones
in the binary representation of a sum of positive integers.)  This
together with the previous observation proves our claim.
\end{proof}
By combining Claims \ref{cl:vardy} and \ref{cl:2dpl2} we obtain the
proof of Theorem \ref{thm:CdC} as follows. Consider an instance of
Problem \ref{prob:sum} with $n=2d+2$ and let $D=d+1$. Claim
\ref{cl:vardy} states that this instance has a solution if and only if
the $(d+1)\times (n+1)=D\times (2D+1)$ matrix defined in the claim has
$D$ linearly dependent columns, which must be NP-hard to decide by
Claim \ref{cl:2dpl2}.
\end{proof}

\begin{cor}\label{cor:CutdCut}
The problem of deciding whether an (explicitly given) linear matroid
contains two disjoint cuts is NP-complete.
\end{cor}
\begin{proof}
Since the dual matroid of the linear matroid is also linear, and we
can construct a representation of this dual matroid from the
representation of the original matroid, this problem is equivalent
to the problem of deciding whether a linear matroid contains two
disjoint circuits, which is NP-complete by Theorem \ref{thm:CdC}.
\end{proof}

\bibliographystyle{amsplain} \bibliography{bpath}

\end{document}